\documentclass[12pt]{article}
\usepackage{amssymb,amsmath,latexsym,color}
\usepackage[mathscr]{eucal}
\usepackage{geometry}
\usepackage[auth-sc]{authblk}

\newtheorem{theorem}{Theorem}[section]
\newtheorem{lemma}[theorem]{Lemma}

\newtheorem{remark}[theorem]{Remark}
\newenvironment{proof}{\noindent
  \textbf{Proof.}}{\hfill$\Box$\\}

\newcommand{\cut}[1]{}

\newcommand{\instr}[5]{\ensuremath{\hbox to 60 pt
    {${#1}$\hfil${#2}$\hfil$ \rightarrow
      $\hfil${#3}$\hfil${#4}$\hfil${#5}$}}}

\newcommand{\vp}{\ensuremath{\varphi}}

\newcommand{\rel}[1]{\ensuremath{\mathcal{#1}}}

\newcommand{\union}{\, \cup \,}

\newcommand{\inter}{\, \cap \,}

\newcommand{\con}{\wedge}

\newcommand{\imp}{\rightarrow}

\newcommand{\equivalence}{\leftrightarrow}
\newcommand{\bottom}{\perp}

\newcommand{\Var}{\textit{\texttt{Var}}}

\newcommand{\truth}{\ensuremath{\top}}
\newcommand{\falsehood}{\ensuremath{\bot}}

\newcommand{\cm}[1]{\ensuremath{\sf{#1}}}

\newcommand{\mmodel}[1]{\ensuremath{\frak{#1}}}
\newcommand{\states}[1]{\ensuremath{\mathcal{#1}}}

\newcommand{\sat}[3]{\ensuremath{\frak{#1}, #2 \models #3}}

\newcommand{\notsat}[3]{\ensuremath{\mmodel{#1}, #2 \not\models #3}}

\newcommand{\diam}[1]{\ensuremath{\langle #1  \rangle}\,}
\newcommand{\boxm}[1]{\ensuremath{[\, #1\, ]\, }}

\newcommand{\comp}{\, ;\,}
\newcommand{\choice}{\cup}

\newcommand{\psmodel}[1]{\ensuremath{\mathcal{M}^*}}
\newcommand{\pseudomodel}[1]{\ensuremath{\mathcal{M}^{**}}}

\newcommand{\sameas}{\ensuremath{\leftrightharpoons}}

\providecommand{\keywords}[1]{\textsc{\textsc{Keywords:}} #1}

\begin{document}

\title{Complexity and expressivity of propositional dynamic logics
  with finitely many variables\footnote{Prefinal version of the paper
    published in \textit{Logic Journal of the IGPL}, 26(5), 2018,
    pp.539--547. DOI
    \texttt{https://doi.org/10.1093/jigpal/jzy014}}}
\author[1]{Mikhail Rybakov} \author[2]{Dmitry Shkatov} \affil[1]{Tver
  State University and University of the Witwatersrand, Johannesburg,
  \texttt{m\_rybakov@mail.ru}} \affil[2]{University of the
  Witwatersrand, Johannesburg, \texttt{shkatov@gmail.com}} \date{}
\maketitle

\begin{abstract}
  We investigate the complexity of satisfiability for finite-variable
  fragments of propositional dynamic logics.  We consider three
  formalisms belonging to three representative complexity classes,
  broadly understood,---regular PDL, which is EXPTIME-complete, PDL
  with intersection, which is 2EXPTIME-complete, and PDL with parallel
  composition, which is undecidable.  We show that, for each of these
  logics, the complexity of satisfiability remains unchanged even if
  we only allow as inputs formulas built solely out of propositional
  constants, i.e. without propositional variables.  Moreover, we show
  that this is a consequence of the richness of the expressive power
  of variable-free fragments: for all the logics we consider, such
  fragments are as semantically expressive as entire logics.  We
  conjecture that this is representative of PDL-style, as well as
  closely related, logics.
\end{abstract}


\section{Introduction}

The propositional dynamic logic, PDL, introduced in~\cite{FL79}, has
ever since been used for reasoning about the input-output behaviour of
terminating programs.  Over the years, it has been extended in various
ways to deal with a wider variety of terminating programs
\cite{Streett82, HS82, Vardi85, Peleg87, Goldblatt92}.  Also, various
formalisms closely linked to PDL have been developed for applications
in areas other than reasoning about programs; among them are knowledge
representation \cite{GL94,DLNN97,Massacci01}, querying semistructured
data \cite{ADdeR03}, data analysis \cite{CO85}, and linguistics
\cite{Kracht95}.

Clearly, the complexity of satisfiability---equivalently,
validity---problem for all of these variants of PDL is of crucial
importance to their applications in the above-mentioned domains.
Typically, for formulas containing an arbitrary number of
propositional variables, the complexity of satisfiability problem for
variants of PDL is rather high: it ranges from
EXPTIME-complete~\cite{FL79} to undecidable~\cite{TinBal2014}.

It has, however, been observed that, in practice, one rarely uses
formulas containing a large number of propositional
variables---usually, this number is rather small.  This raises the
question of whether the complexity of satisfiability for PDL can be
tamed by restricting the language to a finite number of propositional
variables.  Such an effect is not, after all, entirely unknown: for
many logics, the complexity of satisfiability goes down from
``intractable'' to ``tractable'' once we place a limit on the number
of propositional variables that can be used in the construction of
formulas.  For the classical propositional logic, as well as for the
normal extensions of the modal logic {\bf K5}~\cite{NTh75}, which
include logics {\bf K45}, {\bf KD45}, and {\bf S5} (see
also~\cite{Halpern95}), the complexity of satisfiability goes down
from NP-complete to polynomial-time computable once we restrict the
number of propositional variables to any finite number.  Similarly, as
follows from~\cite{Nishimura60}, the complexity of satisfiability for
intuitionistic propositional logic goes down from PSPACE-complete to
polynomial-time computable if we consider only formulas of one
variable.

The main contribution of the present paper is to show that for
propositional dynamic logics this route to reducing the complexity of
satisfiability seems to be closed: even formulas built out of
propositional constants, and thus containing no propositional
variables at all, are as hard to test for satisfiability as formulas
with an arbitrary number of propositional variables.  

We suspect that this behaviour is representative of PDL-style logics.
It would, however, be difficult to make an exhaustive case, given a
wild proliferation of such formalisms.  What we do instead is pick
three examples that are representative in the sense of their
satisfiability problems belonging to three representative complexity
classes (broadly understood, i.e., treating ``undecidable'' as a
complexity class); namely, we consider regular PDL, which has an
EXPTIME-complete satisfiability problem~\cite{FL79}, PDL with
intersection, which has a 2EXPTIME-complete satisfiability
problem~\cite{LL05}, and PDL with parallel composition, which has an
undecidable satisfiability problem~\cite{TinBal2014}.  We show that
satisfiability problem for the variable-free fragment of each of these
logics is as hard as for the entire logic.  Moreover, we show that
this is a consequence of the richness of the expressive power of
variable-free fragments: for all the logics we consider, variable-free
fragments are as semantically expressive as entire logics.

Similar results for other propositional modal logics have been
obtained in~\cite{BS93}, \cite{Halpern95}, \cite{Hem01}, \cite{DS02},
\cite{Sve03}, and \cite{ChRyb03}.  The techniques used in those
studies are not directly applicable to obtain the results presented in
this paper; we do, however, substantially draw on the ideas
from~\cite{BS93} and~\cite{Halpern95}.

The paper is organised as follows. In section~\ref{sec:pdl}, we recall
the syntax and semantics of the logics we consider.  Then, in
section~\ref{sec:finite-variable-fragments}, we present our results
about complexity and expressivity of their variable-free fragments.
We conclude in section~\ref{sec:conclusion}.

\section{Syntax and semantics}
\label{sec:pdl}

In this section, we recall the syntax and semantics of PDL with
intersection ({\bf IPDL}), regular PDL ({\bf PDL}), and PDL with
parallel composition ({\bf PRSPDL}).

The language of {\bf IPDL} contains a countable set
$\Var = \{p_1, p_2, \ldots \}$ of propositional variables, the
propositional constant $\bottom$ (``falsehood''), the Boolean
connective $\imp$, and modalities of the form $\boxm{\alpha}$, where
$\alpha$ ranges over program terms built out of a countable set
$AP = \{a_1, a_2, \ldots\}$ of atomic program terms as well as
formulas, using the operations ?  (test), $\comp$ (composition),
$\choice$ (choice), $\inter$ (intersection), and $^\ast$ (iteration).
The intended meaning of the formula $\boxm{\alpha} \varphi$ is that
every execution of the program $\alpha$ at the current state results
in a state where $\varphi$ holds.  Formulas $\varphi$ and program
terms $\alpha$ are simultaneously defined by the following BNF
expressions:
\[\vp := p  \mid \falsehood \mid  (\vp
\imp \vp) \mid \boxm{\alpha} \vp,\]
\[\alpha := a \mid \vp? \mid (\alpha \comp \alpha) \mid (\alpha \choice \alpha)
\mid (\alpha \inter \alpha) \mid \alpha^*, \]
where $p$ ranges over \Var\ and $a$ ranges over $AP$.  The other
connectives are defined as usual.  Formulas are evaluated in Kripke
models.  A Kripke model is a tuple
$\mmodel{M} = (\states{S}, \{\rel{R}_a\}_{a \in AP}, V)$, where
\states{S} is a non-empty set (of states), $\rel{R}_a$ is a binary
(accessibility) relation on \states{S}, and $V$ is a (valuation)
function $V: \Var \rightarrow 2^{\states{S}}$.  Accessibility
relations for non-atomic program terms as well as the satisfaction
relation between models, states, and formulas are defined by
simultaneous induction as follows:
\begin{itemize}
\item $(s, t) \in \rel{R}_{\vp?}$ \sameas\ $s = t$ and \sat{M}{s}{\vp};  
  \nopagebreak[3]
\item $(s, t) \in \rel{R}_{\alpha \comp \beta}$ \sameas\ $(s, u) \in
  \rel{R}_\alpha$ and $(u, t) \in \rel{R}_\beta$, for some $u \in
  \states{S}$; \nopagebreak[3]
\item $(s, t) \in \rel{R}_{\alpha \choice \beta}$ \sameas\ $(s, t) \in
  \rel{R}_\alpha$ or $(s, t) \in \rel{R}_\beta$;
  \nopagebreak[3]
\item $(s, t) \in \rel{R}_{\alpha \inter \beta}$ \sameas\ $(s, t) \in
  \rel{R}_\alpha$ and $(s, t) \in \rel{R}_\beta$;
  \nopagebreak[3]
\item $(s, t) \in \rel{R}_{\alpha^*}$ \sameas\
  $(s, t) \in \rel{R}^\ast_{\alpha}$, where $\rel{R}^\ast_{\alpha}$ is
  the reflexive, transitive closure of $\rel{R}_{\alpha}$;
   \nopagebreak[3]
  \item \sat{M}{s}{p_i} \sameas\ $s \in V(p_i)$;
\nopagebreak[3]
\item \sat{M}{s}{\falsehood} never holds;
\nopagebreak[3]
\item \sat{M}{s}{\vp \imp \psi} \sameas\ \sat{M}{s}{\vp} implies
  \sat{M}{s}{\psi};
\nopagebreak[3]
\item \sat{M}{s}{\boxm{\alpha} \vp} \sameas\ \sat{M}{t}{\vp} whenever
  $(s, t) \in \rel{R}_{\alpha}$.
\end{itemize}
A formula is satisfiable if it is satisfied at some state of some
model. A formula is valid if it is satisfied by every state of every
model. Formally, by {\bf IPDL}, we mean the set of all valid formulas
in this language.

The language of {\bf PDL} differs from that of {\bf IPDL} in that it
does not contain program operations $\inter$ and ?.  The semantics is
modified accordingly.

The language of {\bf PRSPDL} is interpreted on models made up of
states possessing inner structure: a state $s$ is a composition
$x \ast y$ of states $x$ and $y$ if $s$ can be separated into
components $x$ and $y$; in general, there is no requirement that,
given states $x$ and $y$, a composition $x \ast y$ is a unique
state. The program terms are formed out of atomic program terms as
well as four special program terms $r_1$, $r_2$ (recovery of the first
and second \mbox{$\ast$-com\-po\-nents}, respectively, of a state),
$s_1$, and $s_2$ (storing a state as the first and second
\mbox{$\ast$-com\-po\-nents}, respectively, of a composite state),
using the operations ? (test), ${}^\ast$ (iteration), and $||$
(parallel composition).  Note that the language of {\bf PRSPDL} does
not contain the operation of union of program terms.  A Kripke model
is a tuple
$\mmodel{M} = (\states{S}, \{\rel{R}_a\}_{a \in AP}, \ast, V)$, where
\states{S}, $\rel{R}_a$, and $V$ have the same meaning as in Kripke
models for {\bf IPDL}, and $\ast$ is a function
$\states{S} \times \states{S} \to 2^{\states{S}}$.  The meaning of
$||$, $r_1$, $r_2$, $s_1$, and $s_2$ is given by the following
clauses:
\begin{itemize}
\item $(s, t) \in \rel{R}_{\alpha\, ||\, \beta}$ \sameas\ there exist
  $x_1, y_1, x_2, y_2 \in \states{S}$ such that $s \in x_1 \ast x_2$,
  $t \in y_1 \ast y_2$, $(x_1, y_1) \in \rel{R}_\alpha$, and
  $(x_2, y_2) \in \rel{R}_\beta$; \nopagebreak[3]
\item $(s, t) \in \rel{R}_{r_1}$ \sameas\ there exists $u \in
  \states{S}$ such that $s \in t \ast u$;
  \nopagebreak[3]
\item $(s, t) \in \rel{R}_{r_2}$ \sameas\ there exists $u \in
  \states{S}$ such that $s \in u \ast t$;
  \nopagebreak[3]
\item $(s, t) \in \rel{R}_{s_1}$ \sameas\ there exists $u \in
  \states{S}$ such that $t \in s \ast u$;
  \nopagebreak[3]
\item $(s, t) \in \rel{R}_{s_2}$ \sameas\ there exists $u \in
  \states{S}$ such that $t \in u \ast s$.
\end{itemize}
The models thus defined are referred to in~\cite{TinBal2014} as
``\mbox{$\ast$-se\-pa\-ra\-ted}.'' The authors of~\cite{TinBal2014} consider a
number of logics in the same language, which differ in the conditions
placed on the function $\ast$ in their semantics. For our purposes,
it suffices to consider only one of the logics
from~\cite{TinBal2014},---the rest can be dealt with in a similar way.

The notions of satisfiability and validity are defined as for {\bf
  IPDL} and {\bf PDL}. 

For each of the logics we consider, by a variable-free fragment we
mean the subset of the logic containing only variable-free
formulas---i.e., formulas not containing any propositional variables.
Given formulas $\vp$, $\psi$ and a propositional variable $p$, we
denote by $\vp(p/\psi)$ the result of uniformly substituting $\psi$
for $p$ in $\vp$.

\section{Finite-variable fragments}
\label{sec:finite-variable-fragments}

In this section, we show that variable-free fragments of {\bf IPDL},
{\bf PDL}, and {\bf PRSPDL} have the same expressive power and
computational complexity as the entire logics, by embedding each logic
into its variable-free fragment; in the case of {\bf IPDL} and {\bf
  PDL}, the embeddings are polynomial-time computable.  We initially
work with {\bf IPDL} and subsequently point out how that work carries
over to {\bf PDL} and {\bf PRSPDL}.

Let $\vp$ be an arbitrary {\bf IPDL}-formula.  Assume that $\vp$ only
contains propositional variables $p_1, \ldots, p_n$ and atomic program
terms $a_1, \ldots, a_l$.  Let
$\gamma = a_1 \union \ldots \union a_l$. First, recursively define
translation $\cdot'$ as follows:
\begin{center}
  \begin{tabular}{llll}
    ${a_j}'$ & = & $a_j,$ $\mbox{~~where~} j \in \{1, \ldots, l \}$; \\
    $(\alpha \comp \beta)'$ & = & $\alpha' \comp \beta'$; & \\
    $(\alpha \choice \beta)'$ & = & $\alpha' \choice \beta'$; & \\
    $(\alpha \inter \beta)'$ & = & $\alpha' \inter \beta'$; & \\
    $(\alpha^*)'$ & = & $(\alpha')^*$; & \\
    $(\phi?)'$ & = & $(\phi')?$; & \\
    ${p_i}'$ & = & $p_i,$ $\mbox{~~where~} i \in \{1, \ldots, n \}$; \\
    $(\bottom)'$ & = & $\bottom$; & \\
    $(\phi \imp \psi)'$ & = & $\phi' \imp \psi'$; & \\
    $(\boxm{\alpha} \phi)'$ & = & $\boxm{\alpha'} (p_{n+1} \imp \phi')$. & 
\end{tabular}
\end{center}
Second, define
$$ \Theta = p_{n+1} \con \boxm{\gamma^*} (\diam{\gamma} p_{n+1} \imp
p_{n+1}). $$
Finally, let $$\widehat{\vp} = \Theta \con \vp'. $$

\begin{lemma}
  \label{lem:varphi-truth}
  Formula $\varphi$ is satisfiable if, and only if, formula
  $\widehat{\vp}$ is satisfiable.
\end{lemma}

\begin{proof}
  Suppose $\widehat{\vp}$ is not satisfiable.  Then,
  $\neg \widehat{\vp} \in {\bf IPDL}$ and, since {\bf IPDL} is closed
  under substitution,
  $\neg \widehat{\vp} (p_{n+1} / \top) \in {\bf IPDL}$.  As
  $\widehat{\vp} (p_{n+1} / \top) \equivalence \vp \in {\bf IPDL}$, we
  have $\neg \vp \in {\bf IPDL}$; thus, $\vp$ is not satisfiable.

  Suppose that $\widehat{\vp}$ is satisfiable. In particular, let
  $\sat{M}{s_0}{\widehat{\vp}}$ for some model $\mmodel{M}$ and some
  $s_0$ in $\mmodel{M}$. Define $\mmodel{M}'$ to be the smallest
  submodel of $\mmodel{M}$ such that
  \begin{itemize}
  \item $s_0$ is in $\mmodel{M'}$;
  \item if $x$ is in $\mmodel{M'}$, $x \rel{R}_{\gamma} y$, and
    \sat{M}{y}{p_{n+1}}, then $y$ is also in $\mmodel{M'}$.
  \end{itemize}
  Notice that $p_{n+1}$ is universally true in $\mmodel{M}'$.  It is
  straightforward to show that, for every subformula $\psi$ of $\vp$
  and every $s$ in $\mmodel{M}'$, we have $\sat{M}{s}{\psi'}$ if, and
  only if, $\sat{M'}{s}{\psi}$.  As $\sat{M}{s_0}{\vp'}$, this gives
  us $\sat{M'}{s_0}{\vp}$; hence, $\vp$ is satisfiable.
\end{proof}

\begin{remark}
  \label{remark:pn+1}
  It follows from the proof of Lemma~\ref{lem:varphi-truth} that, if
  $\widehat{\vp}$ is satisfiable, then it is satisfiable in a model
  where $p_{n+1}$ is universally true.  Indeed, if $\widehat{\vp}$ is
  satisfiable, then $\vp$ is satisfiable in a model where $p_{n+1}$ is
  universally true. The claim follows from the fact that $\vp$ is
  equivalent to $\widehat{\vp} (p_{n+1} / \top)$.
\end{remark}

Now, consider the following class $\cm{M}$ of finite models.  Let $b$
be the lexicographically first atomic program term of $\varphi$ if
$\varphi$ contains such terms; otherwise, let $b$ be $a_1$.  For every
$m \in \{1, \ldots, n + 1\}$, where $p_1, \ldots, p_n$ are the
variables in $\varphi$, class $\cm{M}$ contains a unique member
$\mmodel{M}_m$, defined as follows:
$\mmodel{M}_m = (\states{S}_m, \{\rel{R}_a\}_{a \in AP}, V_m)$, where
  \begin{itemize}
  \item $\states{S}_m = \{r_m, t^m, s_1^m, s_2^m, \ldots, s_{m}^m\}$;
  \item $\rel{R}_{b}$ is the transitive closure of the relation $\{
    \langle r_m, t^m \rangle, \langle t^m, t^m \rangle, \langle r_m,
    s_1^m \rangle \} \linebreak \union \{ \langle s_i^m, s_{i+1}^m
    \rangle : 1 \leq i \leq m - 1 \}$;
  \item $\rel{R}_{a} = \varnothing$ if $a \neq b$;
  \item $V_m (p) = \varnothing$ for every $p \in \Var$.
  \end{itemize}

  The model $\mmodel{M}_m$ is depicted in Figure~\ref{fig_Mm}, where
  arrows represent $\rel{R}_{b}$; to avoid clutter, arrows are omitted
  whenever the presence of $\rel{R}_{b}$ can be deduced from its
  transitivity; the circle represents a state related by $\rel{R}_{b}$
  to itself, and solid dots represent states without such loops.

  \begin{figure}
\begin{center}
\begin{picture}(100,195)
\put(59.5,10){$\bullet$}
\put(59.5,50){$\bullet$}
\put(59.5,90){$\bullet$}
\put(59.5,172){$\bullet$}
\put(20,50){$\circ$}
\put(70,10){$r_m\models A_m$}
\put(45,10){}
\put(45,10){}
\put(42,50){}
\put(45,90){}
\put(42,130){}
\put(45,190){}
\put(2,50){}
\put(70,10){}
\put(70,50){$s_1^m$}
\put(70,90){$s_2^m$}
\put(70,170){$s_{m}^m$}
\put(20,60){$t^m$}
\put(62.1,14){\vector(0,1){36}}
\put(62.1,54){\vector(0,1){36}}
\put(62.1,94){\vector(0,1){36}}
\put(60,14){\vector(-1,1){37}}
\put(62.1,152){\vector(0,1){19}}
\put(60.1,138){$\vdots$}
\end{picture}
\caption{Model $\frak{M}_m$}
\label{fig_Mm}
\end{center}
\end{figure}

  We now define formulas that will be true at the roots of models from
  $\cm{M}$. For $j \geqslant 0$, inductively define the formula
  $\langle b \rangle^j \psi$ as follows:
  $\langle b \rangle^0 \psi = \psi$;
  $\langle b \rangle^{k+1} \psi = \diam{b} \langle b \rangle^k \psi$.
  Next, for every $m \in \{1, \ldots, n + 1\}$, define
$$
A_m = \langle b \rangle^m \boxm{b} \falsehood \con \neg
\langle b \rangle^{m+1} \boxm{b} \falsehood \con \diam{b}
(\diam{b} \truth \con \boxm{b} \diam{b} \truth).
$$
\begin{lemma}
  \label{lem:alphas}
  Let $\mmodel{M}_k \in \cm{M}$ and let $x$ be a state in
  $\mmodel{M}_k$. Then, $\mmodel{M}_k, x \models A_m$ if, and only if,
  $k = m$ and $x = r_m$.
\end{lemma}

\begin{proof}
  Straightforward.
\end{proof}

Now, define $$B_m = \diam{b} A_m.$$

Let $\sigma$ be a (substitution) function that, given an {\bf
  IPDL}-formula $\psi$, replaces all occurrences of $p_i$ in $\psi$ by
$B_i$, where $1 \leqslant i \leqslant n + 1$. Finally, define
$$\varphi^* = \sigma(\widehat{\vp})$$
to produce a variable-free formula $\varphi^*$. 
\begin{lemma}
  \label{lem:main_lemma}
  Formula $\varphi$ is satisfiable if, and only if, formula
  $\varphi^*$ is satisfiable.
\end{lemma}

\begin{proof}
  Suppose that $\varphi$ is not satisfiable.  Then, by
  Lemma~\ref{lem:varphi-truth}, $\widehat{\vp}$ is not
  satisfiable, either, and hence
  $\neg \widehat{\vp} \in {\bf IPDL}$.  Since {\bf IPDL} is closed
  under substitution, $\neg \varphi^* \in {\bf IPDL}$ and, thus,
  $\vp^*$ is not satisfiable.

  Suppose that $\vp$ is satisfiable.  Then, in view of
  Lemma~\ref{lem:varphi-truth} and Remark~\ref{remark:pn+1},
  $\sat{M}{s_0}{\widehat{\vp}}$ for some \mmodel{M} such that
  $p_{n+1}$ is true at every state of \mmodel{M} and some $s_0$ in
  \mmodel{M}.  Define model \mmodel{M'} as follows. Attach to
  \mmodel{M} all the models from \cm{M}; then, for every $x$ in
  $\mmodel{M}$, put $x \rel{R}_{b} r_m$ (where $r_m$ is the root of
  $\mmodel{M}_m \in \cm{M}$) exactly when \sat{M}{x}{p_m}.  Notice
  that $r_{n+1}$ is accessible in \mmodel{M'} from every $x$ in
  $\mmodel{M}$.

  To conclude the proof, it suffices to show that
  \sat{M'}{s_0}{\varphi^*}.  It is easy to check that
  \sat{M'}{s_0}{\sigma(\Theta)}.  It then remains to show that
  \sat{M'}{s_0}{\sigma(\vp')}.  To that end, it suffices to show that
  \sat{M}{x}{\psi'} if, and only if, \sat{M'}{x}{\sigma (\psi')}, for
  every subformula $\psi$ of $\varphi$ and every $x$ in $\mmodel{M}$.
  This can be done by induction on $\psi$; we only consider the base
  case, leaving the rest to the reader.

  Let \sat{M'}{x}{B_i}.  Then, for some $y$ in $\mmodel{M}'$, we have
  $x \rel{R}'_{b} y$ and \sat{M'}{y}{A_i}. This is only possible if
  $y$ is not in \mmodel{M}.  Indeed, suppose otherwise.  Then,
  \sat{M'}{y}{p_{n+1}}, and therefore, $y \rel{R}'_{b} r_{n+1}$.
  Hence, \sat{M'}{y}{\langle b \rangle^{i+1} \boxm{b} \falsehood}, and
  therefore, \notsat{M'}{y}{A_i}, resulting in a contradiction. Thus,
  $y$ is in $\mmodel{M}_m$, for some $m \in \{1, \ldots, n +
  1\}$.
  Then, by Lemma~\ref{lem:alphas}, $y = r_i$, and therefore, by
  definition of $\mmodel{M}'$, we have \sat{M}{x}{p_i}. The other
  direction is straightforward.
\end{proof}

\begin{theorem}
  \label{thr:ipdl}
  There exists a mapping that embeds {\bf IPDL} into its variable-free
  fragment in polynomial time.
\end{theorem}

We now look at the complexity-theoretic implications of
Theorem~\ref{thr:ipdl}.  It has been shown in~\cite{LL05} that the
fragment of {\bf IPDL} containing a single atomic program term is
2EXPTIME-complete.  This gives us the following:

\begin{theorem}
  The satisfiability problem for the fragment of {\bf IPDL} containing
  variable-free formulas with a single atomic program term is {\rm
    2EXPTIME}-complete.
\end{theorem}

We now point out how the work we have done so far for {\bf IPDL}
carries over to {\bf PDL} and {\bf PRSPDL}.

It is easy to check that the construction presented above works for
{\bf PDL}, as well, if we omit the details peculiar to {\bf IPDL}.
This gives us the following:

\begin{theorem}
  There exists a mapping that embeds {\bf PDL} into its variable-free
  fragment in polynomial time.
\end{theorem}

Since satisfiability problem for {\bf PDL} with a single atomic
program term is EXPTIME-complete~\cite{FL79}, we have the following:

\begin{theorem}
  \label{cor:pdl}
  The satisfiability problem for the fragment of {\bf PDL} containing
  variable-free formulas with a single atomic program term is {\rm
    EXPTIME}-complete.
\end{theorem}

We next show how to modify the above argument for {\bf PRSPDL}. We
remind the reader that we confine our attention to {\bf PRSPDL} over
\mbox{$\ast$-se\-pa\-ra\-ted} models; other variants of this formalism
considered in~\cite{TinBal2014} can be treated in essentially the same
way.  We first need to construct the analogue of formula
$\widehat{\vp}$.  It is straightforward to define the translation
$\cdot'$:
\begin{center}
  \begin{tabular}{lcl}
    ${a_i}'$ & = & $a_i$, $\mbox{~~where~} i \in \{1, \ldots, l \}$; \\
    ${r_i}'$ & = & $r_i$, $\mbox{~~where~} i \in \{1, 2 \}$; \\
    ${s_i}'$ & = & $s_i$, $\mbox{~~where~} i \in \{1, 2 \}$; \\
    $(\alpha \comp \beta)'$ & = & $\alpha' \comp \beta'$;  \\
    $(\alpha\, ||\, \beta)'$ & = & $\alpha'\, ||\, \beta'$  \\
    $(\alpha^*)'$ & = & $(\alpha')^*$;  \\
    $(\phi?)'$ & = & $(\phi')?$;  \\
    ${p_i}'$ & = & $p_i$, $\mbox{~~where~} i \in \{1, \ldots, n \}$; \\
    $(\bottom)'$ & = & $\bottom'$;  \\
    $(\phi \imp \psi)'$ & = & $\phi' \imp \psi'$;  \\
    $(\boxm{\alpha} \phi)'$ & = & $\boxm{\alpha'} (p_{n+1} \imp \phi')$. 
\end{tabular}
\end{center}
We next define the analogue of formula $\Theta$. As {\bf PRSPDL} does
not have the operation of choice on program terms, we proceed as
follows.  Let
\begin{center}
  $\alpha^1_1 \dots \alpha^1_{n_1}$ \\
  \dots \\
  $\alpha^k_1 \dots \alpha^k_{n_k}$
\end{center}
be all sequences of nested program terms in $\varphi$.  Then,
$$ \Theta = p_{n+1} \con \bigwedge_{i=1}^k \bigwedge_{j=1}^{n_k - 1 }
\boxm{\alpha^i_1} \dots \boxm{\alpha^i_j} (\diam{\alpha^i_{j+1}}
p_{n+1} \imp p_{n+1}). $$
Finally, let $$ \widehat{\vp} = \Theta \con \varphi'.$$ From here on,
we argue exactly as in the case of {\bf IPDL} to obtain the following:
\begin{theorem}
  There exists a mapping that embeds {\bf PRSPDL} over
  \mbox{$\ast$-se\-pa\-ra\-ted} models into its variable-free
  fragment.
\end{theorem}
\begin{theorem}
  The variable-free fragment of {\bf PRSPDL} over
  \mbox{$\ast$-se\-pa\-ra\-ted} models is undecidable.
\end{theorem}

\begin{remark}
  It is well-known that the consequence relation for propositional
  dynamic logics is not compact, as the formula $\boxm{a^\ast} \vp$
  follows from the infinite set $\{\boxm{a}^n \vp : n \geqslant 0\}$
  of formulas but not from any of its finite subsets; thus, the
  consequence relation is not reducible to satisfiability for
  formulas.  The technique presented above can be used to reduce the
  consequence relation for the logics we have considered to the
  consequence relation for their variable-free fragments.  To that
  end, unless the number of propositional variables occurring in the
  premises is finite, we need to use an extra atomic program term
  corresponding to the accessibility relation connecting the roots of
  the models attached in the proof of Lemma~\ref{lem:main_lemma} to
  the original model.  This is necessary as in the proof of
  Lemma~\ref{lem:main_lemma} we relied on the variable $p_{n+1}$, used
  as a marker of the worlds of the original model, having the maximal
  index of all the variables of the formula $\vp$.
\end{remark}

\section{Conclusion}
\label{sec:conclusion}

We have shown that for three variants of propositional dynamic logic
representative of various complexity classes, broadly understood, the
complexity of satisfiability remains the same if we restrict the
language to formulas built out of propositional constants, i.e.,
without the use of propositional variables.  This is a consequence of
the richness of the expressive power of the variable-free
fragments---as we have shown, they are as expressive as the logics
with an infinite supply of propositional variables.

We suspect that these results are representative of how PDL-style
formalisms behave.  If this is indeed so, the important question for
future research is to find out if there are ways to tame the
complexity of satisfiability that might be applicable en masse to a
wide range of PDL-style logics and that might be of relevance to how
these formalisms are applied in practice.

\section*{Acknowledgements}

We thank the two anonymous referees for their remarks, which have
helped to substantially improve the presentation of the paper. In
particular, the paper has become shorter and more focused.

\end{document}